\def\MINE{0}    
\def\DRAFT{0}   

\if\MINE1
    \documentclass[11]{article}
    \usepackage{amsthm}
\else
    \documentclass{llncs}   
\fi

\usepackage{latexsym,amssymb,amsmath}
\usepackage{graphics,graphicx}
\usepackage{ifthen}
\usepackage{ifthen}
\usepackage{xspace}
\usepackage{paralist}  %
\usepackage{multirow}  
\usepackage{shuffle}   %
\usepackage{environ}   %
\usepackage[colorlinks,citecolor=blue]{hyperref}
\usepackage[usenames,dvipsnames]{xcolor}

\usepackage{subcaption} 

\usepackage{marginnote} 
\usepackage{verbatim}  
\usepackage{fp}

\usepackage{bm} 

\newcommand{\drawasxy}[8]{
\def\aa{#5}
\def\sa{#6}
\def\xa{{cos(\aa)*\sa/(abs(cos(\aa))+abs(sin(\aa)))}}
\def\ya{{sin(\aa)*\sa/(abs(cos(\aa))+abs(sin(\aa)))}}

\def\ab{#7}
\def\sb{#8}
\def\xb{{cos(\ab)*\sb/(abs(cos(\ab))+abs(sin(\ab)))}}
\def\yb{{sin(\ab)*\sb/(abs(cos(\ab))+abs(sin(\ab)))}}

\draw [to path={
.. controls +(\xa,\ya) and +(\xb, \yb) .. (\tikztotarget) \tikztonodes}]
 [->] (#1) to
node[#4] {#3} (#2);
}

\newcommand{\drawJurBlock}[7]{
\def\stavNad{#1}
\def\stavA{#2}
\def\stavAlabel{#3}
\def\stavB{#4}
\def\stavBlabel{#5}
\def\XDifWidth{#6}
\def\YDifHeight{#7}

\node[state,draw=none,inner sep=7pt,minimum size=0pt] [below of=\stavNad] (\stavA) [label=center: \stavAlabel] {};
\node[state,draw=none,inner sep=7pt,minimum size=0pt] [right of=\stavA] (\stavB) [label=center: \stavBlabel] {};
\draw ($ (\stavA) +(-1*\XDifWidth,-1*\YDifHeight)$)
-- ($ (\stavB) +(1*\XDifWidth,-1*\YDifHeight)$)
-- ($ (\stavB) +(1*\XDifWidth,1*\YDifHeight)$)
-- ($ (\stavA) +(-1*\XDifWidth,1*\YDifHeight)$)
-- cycle;
}

\newcommand\pssi{\par\smallskip\indent}
\newcommand\pmsi{\par\medskip\indent}

\newcommand\pssn{\par\smallskip\noindent}
\newcommand\pmsn{\par\medskip\noindent}
\newcommand\pbsn{\par\bigskip\noindent}
\newcommand\pnsi{\par\indent}
\newcommand\pnsn{\par\noindent}

\newcommand\emdef[1]{\textsf{#1}}  

\newcommand{\NOTE}[1]{\if\DRAFT1\marginnote{\footnotesize\color{red}#1}\fi}
\newcommand{\sklabel}[1]{\label{#1}\if\DRAFT1\NOTE{#1}\fi}
\newcommand{\sk}[2]{\textit{{\color{blue}\textbf{#1}#2}}} 

\newlength{\algoindent}   
\if\MINE1
\theoremstyle{plain}
\newtheorem{theorem}{Theorem}[]
\newtheorem{proposition}[theorem]{Proposition}
\newtheorem{lemma}[theorem]{Lemma}
\newtheorem{corollary}[theorem]{Corollary}

\theoremstyle{definition}
\newtheorem{definition}[theorem]{Definition}
\newtheorem{example}[theorem]{Example}
\theoremstyle{remark}
\newtheorem{remark}[theorem]{Remark}
\fi

\newcommand{\cS}{\ensuremath{\mathcal{S}}\xspace}

\renewcommand{\phi}{\varphi}
\renewcommand{\epsilon}{\varepsilon}
\newcommand{\sse}{\subseteq}

\newcommand{\es}{\emptyset}
\newcommand{\sm}{-}  

\newcommand{\N}{\ensuremath{\mathbb{N}}\xspace}

\newcommand{\Z}{\ensuremath{\mathbb{Z}}\xspace}

\newcommand{\ew}{\ensuremath{\lambda}\xspace}      
\newcommand{\al}{\ensuremath{\Sigma}\xspace}       
\newcommand{\alstar}{\ensuremath{\al^*}\xspace}    

\newcommand{\rel}[1]{\ensuremath\mathrm{R}({#1})\xspace}  
\newcommand{\relb}[1]{\ensuremath\mathrm{R}\big({#1}\big)\xspace}  
\newcommand{\trt}{\ensuremath{\bm{t}}\xspace}   
\newcommand{\trs}{\ensuremath{\bm{s}}\xspace}   
\newcommand{\trtinv}{\ensuremath{\trt^{-1}}\xspace} 
\newcommand{\trsinv}{\ensuremath{\trs^{-1}}\xspace} 
\usepackage{tikz}
\usetikzlibrary{%
  arrows,%
  positioning,%
  decorations.pathmorphing,%
  decorations.pathreplacing,%
  automata,%
  calc,%
}

\mathchardef\mhyphen="2D

\newcommand{\ep}[0]{\bm{\lambda}}  
\newcommand{\ewp}[0]{\ew/\ew}    
\newcommand{\Px}[1]{\mathop{\mathrm{Prefix}}(#1)}
\newcommand{\dom}[0]{\mathrm{dom}}
\newcommand{\walk}[1]{\langle#1\rangle}
\newcommand{\corr}[1]{{\mathop{\mathrm{corr}(#1)}}}
\newcommand{\Corr}[1]{\mathop{\mathrm{Corr}(#1)}}
\newcommand{\win}[0]{w_{\mathrm{in}}}
\newcommand{\wout}[0]{w_{\mathrm{out}}}
\newcommand{\lbl}[1]{\mathop{\mathrm{label}}(#1)}

\newcommand{\Path}[1]{\mathop{\mathrm{Path}}(#1)}
\newcommand{\Comput}[1]{\mathop{\mathrm{Comput}}(#1)}
\newcommand{\aComput}[1]{\mathop{\mathrm{AccComput}}(#1)}

\newcommand{\len}[1]{|#1|}
\newcommand{\prj}[1]{\mathsf{pr}_{#1}}

\newtheorem{construction}[theorem]{Construction}

\begin{document}

\newcommand{\thetitle}{Partitioning a Symmetric Rational Relation \\ into Two Asymmetric Rational Relations}
\newcommand{\theabstract}{We consider the problem of partitioning effectively a given symmetric (and irreflexive) rational relation R into two asymmetric rational relations. This problem is motivated by a recent method of embedding an R-independent language into one that is maximal R-independent, where the method requires to use an asymmetric partition of R. We solve the problem when R is realized by a zero-avoiding transducer (with some bound k): if the absolute value of the input-output length discrepancy of a computation exceeds k then the length discrepancy of the computation cannot become zero. This class of relations properly contains all recognizable, all left synchronous, and all right synchronous relations. We leave the asymmetric partition problem open when R is not realized by a zero-avoiding transducer. We also show examples of  total word-orderings for which there is a relation R that cannot be partitioned into two asymmetric rational relations such that one of them is decreasing with respect to the given word-ordering.}
\newcommand{\thekeywords}{asymmetric relations, transducers, synchronous relations, word orderings}
\if\MINE1
\begin{center}
	\textbf{\Large\thetitle}
\pbsn
{\large Stavros Konstantinidis$^{1}$, Mitja Mastnak$^{1}$,
Juraj \v Sebej$^{1,2}$}
\pbsn
\parbox{0.90\textwidth}{$^1$Department of Mathematics and Computing Science, Saint Mary's University, Halifax, NS, Canada, \texttt{s.konstantinidis@smu.ca, mmastnak@cs.smu.ca}}
\pssn
\parbox{0.90\textwidth}{$^2$Institute of Computer Science, Faculty of Science, P. J. \v Saf\'arik University, Ko\v sice,  Slovakia, \texttt{juraj.sebej@gmail.com}}
\pbsn
\parbox{0.90\textwidth}{
\textbf{Abstract.}  \theabstract
\pssn
\textbf{Keywords.}  \thekeywords
}  
\end{center}
\else   
\title{\thetitle}
\author{Stavros Konstantinidis\inst{1} \and Mitja Mastnak\inst{1} \and Juraj {\v S}ebej\inst{1,2}}

\institute{
Saint Mary's University, Halifax, Nova Scotia, Canada,\\
\email{s.konstantinidis@smu.ca}, 
\email{mmastnak@cs.smu.ca}
\and
Institute of Computer Science, Faculty of Science, P. J. \v Saf\'arik University, Ko\v sice,  Slovakia
\email{juraj.sebej@gmail.com}
}
\maketitle
\begin{abstract}\theabstract\end{abstract}
\begin{keywords}\thekeywords\end{keywords}
\fi      

\section{Introduction}\label{sec:intro}
The abstract already serves as the first paragraph of the introduction.

The structure of the paper is as follows. The next section contains basic concepts about relations, word orderings and transducers. \underline{Section~\ref{sec:theproblem}} contains the mathematical statement of the rational asymmetric partition problem and the motivation for considering this problem. \underline{Section~\ref{sec:multicopies}} presents the concept of a $C$-copy of a transducer \trt, which is another transducer that contains many copies of the states of \trt (one copy for each $c\in C$). A $C$-copy of \trt, for appropriate $C$, allows us to produce two transducers realizing an asymmetric partition of the relation realized by \trt. \underline{Section~\ref{sec:letter}} deals with the simple case where the transducer is letter-to-letter (Proposition~\ref{P:ltl}). \underline{Section~\ref{sec:k:noncrossing}} introduces zero avoiding transducers \trt with some bound $k\in\N_0$ and shows a few basic properties: the minimum $k$ is less than the number of states of \trt (Proposition~\ref{P:mink}); every left synchronous and every right synchronous relation is realized by some zero-avoiding transducer with  bound 0 (Prposition~\ref{P:lsyncVSzavoid}). \underline{Section~\ref{sec:k:noncrossing:partition}} shows a construction, from a given input-altering transducer \trs, that produces a certain $C$-copy $\alpha(\trs)$ of \trs such that  $\alpha(\trs)$ realizes the set of all pairs in $\rel{\trs}$ for which the input is greater than the output with respect to the radix total order of words (Theorem~\ref{th:qs}). This construction solves the rational asymmetric partition problem when the given relation is realized by a zero-avoiding transducer. \underline{Section~\ref{sec:variations}} discusses a variation of the problem, where we have a certain fixed total word ordering $[>]$ (e.g., the radix one) and we want to know whether there is a rational symmetric relation $S$ such that not both of $S\cap[>]$ and $S\cap[<]$ are rational (Proposition~\ref{P:variation}). This section also offers as an open problem the general rational asymmetric partition problem (that is when the given $R$ is not realized by a zero-avoiding transducer). The last section contains a few concluding remarks.

\section{Basic Terminology and Notation}\label{sec:notation}
We assume the reader is familiar with  basic concepts of formal languages: alphabet, words (or strings), empty word \ew, language (see e.g., \cite{FLhandbookI,HopcroftUllman}).  We shall use a totally ordered alphabet $\Sigma$; in fact for convenience we assume that $\Sigma=\{0,1,\ldots,q-1\}$, for some integer $q>0$. If a word $w$ is of the form $w=uv$ then $u$ is called a prefix and $v$ is called a suffix of $w$. We shall use $x/y$ to denote the \emdef{pair of words} $x$ and $y$. 
\pnsi 
A \emdef{(binary word) relation} $R$ over $\al$ is a subset of $\al^*\times \al^*$, that is, $R\sse \al^*\times \al^*$. 
We shall use the \pssn
\centerline{infix notation $xRy$ to mean that $x/y\in R$; then, $x/\!\!\!\!Ry$ means $x/y\notin R$.}
\pssn 
The \emdef{domain} $\dom R$ of $R$ is the set $\{x\mid x/y\in R\}$. The inverse $R^{-1}$ of $R$ is the relation $\{y/x\mid x/y\in R\}$.	 
\pmsn
\textbf{Word orderings.} The following types of, and notation about relations over $\al$ are important in this work, where $x,y,z$ are any words in $\al^*$. 
\begin{itemize}
	\item A relation $R$ is called (i) \emdef{irreflexive}, if $x/\!\!\!\!Rx$; (ii) \emdef{reflexive}, if $xRx$ for all $x\in\dom R$; (iii)  \emdef{symmetric}, if $xRy$ implies $yRx$; (iv) \emdef{transitive}, if `$xRy\text{ and }yRz$' implies $xRz$.
	\item A relation $A$ is called \emdef{asymmetric}, if $xAy$ implies $y/\!\!\!\!Ax$. In this case, $A$ must be irreflexive. Moreover we have that $$A\cap A^{-1}=\emptyset \quad\text{ and }\quad A\sse(\al^*\times\al^*)\setminus\{w/w:w\in\al^*\}.$$
	 A \emdef{total asymmetry} is an asymmetric relation $A$ such that either $uAv$ or $vAu$, for all words $u,v$ with $u\neq v$. We shall use the notation `$[>]$' for an arbitrary total asymmetry, as well as the notation `$[>_\alpha]$' for a specific total asymmetry  where $\alpha$ is some identifying subscript. Then, we shall write $u>v$ to indicate that $u/v\in[>]$. Moreover, we shall write   $[<]$ (and $[<_\alpha]$) for the inverse of $[>]$ (and $[<_\alpha]$).
	\item A \emdef{total strict ordering} $[<]$ is a total asymmetry that is also transitive. Examples of total strict orderings are the \emdef{radix} `$[<_r]$' and the \emdef{lexicographic} `$[<_l]$'  ordering.  The lexicographic ordering is the standard dictionary order, for example, $112<_l12<_l3$. The radix ordering is the standard integer ordering when words are viewed as integers and no symbol of $\al$ is interpreted as zero: $3<_r12<_r112$. In both of these orderings, the empty word is the smallest one.
\end{itemize}

\pmsn
\textbf{Pathology.} A path $P$ of a labelled (directed) graph $G=(V,E)$ is a string of consecutive edges, that is, $P\in E^*$ and is of the form $$P=(q_0,\alpha_1,q_1)(q_1,\alpha_2,q_2)\cdots(q_{\ell-1},\alpha_\ell,q_\ell),$$ for some integer $\ell\ge0$, where each $q_i\in V$, each $\alpha_i$ is a label, and each $(q_{i-1},\alpha_i,q_i)\in E$. We shall use the following shorthand notation for that path
\[
P=\walk{q_{i-1},\alpha_i,q_i}_{i=1}^{\ell}.
\]
The \emdef{empty} path is denoted by $\ep$.
We shall \emdef{concatenate} paths in the same way that we concatenate words, provided that the concatenated sequence consists of consecutive edges; thus, $PQ$ is a path of $G$ when  $P,Q$ are paths and the last vertex in $P$ is equal to the first vertex in $Q$. As usual, $P\ep=\ep P=P$, for all paths~$P$. A \emdef{cycle} is a path as above such that $q_0=q_\ell$. The path $P$ \emdef{contains} a cycle $C$, if $C=\walk{q_{j-1},\alpha_j,q_j}_{j=s}^{t}$, for some indices $s,t$, with $1\le s\le t\le\ell$ and $q_{s-1}=q_t$. In this case, we have $P=BCD$ for some paths $C,D$. Moreover, the path $BD$ that results when we remove $C$ from $P$ is well-defined. 
Based on this terminology, we have the following remark.
\begin{remark}
	If $P=BAD$ is a path of $G$, where $A$ is a cycle, and the path $BD$ contains a cycle $C$, then also $P$ contains the cycle~$C$.
\end{remark}
\pnsn
Let $\cS$ be a subset of the cycles contained in $P$. The \emdef{first $\cS$-cycle} of $P$ is the cycle $\walk{q_{j-1},\alpha_j,q_j}_{j=s}^{t}$ of $P$ that has the smallest index $t$ among the cycles of $P$ that are in $\cS$. The path $P-\cS$ is the path that results if we remove from $P$ all $\cS$-cycles: remove the first $\cS$-cycle of $P$ to get a path $P_1$, then the first $\cS$-cycle of $P_1$, etc. 

\pssn \textbf{Transducers.} (\cite{Be:1979,Yu:handbook,Sak:2009})
A transducer is a quintuple\footnote{In general, \trt has an input and an output alphabet, but in our context these are equal.} $\trt=(Q,\al,E,I,F)$ such that $(Q,E)$ is a labelled graph with labels of the form $x/y$, for some $x,y\in\al\cup\{\ew\}$, and $I,F\sse Q$ with $I\neq\es$. The set of vertices $Q$ is also called the set of \emdef{states} of \trt. The set of edges $E$ is also called the set of \emdef{transitions} of \trt. In a transition $e=(p,x/y,q)$ of \trt, $p$ is called the \emdef{source} state of $e$, and $q$ is called the destination state of $e$. The sets $I,F$ are called the initial and final states of \trt, respectively. The \emdef{label} of a path $\walk{q_{i-1},x_i/y_i,q_i}_{i=1}^{\ell}$ is the pair $x_1\cdots x_\ell/y_1\cdots y_\ell$. We write $\lbl{P}$ to denote the label of a path $P$. In particular, $\lbl{\ep}=\ew/\ew$. A \emdef{computation} of \trt is a path $P$ of \trt such that, either $P$ is empty, or the first state of $P$ is in $I$. We write $\Comput{\trt}$ to denote the set of all computations of \trt. The computation $P$ is called \emdef{accepting} if, either $P=\ep$  and $I\cap F\neq\es$, or $P\neq\ep$ and the last state of $P$ is in $F$. We write $\aComput{\trt}$ to denote the set of accepting computations of \trt. The relation \emdef{realized by} \trt  is the set
\[
\rel{\trt}=\{\,\lbl P\mid P\in\aComput{\trt} \}.
\] 
If $\rel{\trt}$ is irreflexive then \trt is called \emdef{input-altering}. If $\rel{\trt}\sse[>]$, for some total asymmetry $[>]$, then \trt is called \emdef{input-decreasing} (with respect to $[>]$). The \emdef{size} $|\trt|$ of \trt is the number of states of \trt plus the number of transitions of \trt. 
\pnsi
If $\trt, \trs$ are transducers then: $\trtinv$ denotes the inverse of \trt such that $\rel{\trtinv}=\rel{\trt}^{-1}$; $\trt\trs$ denotes a transducer such that $\rel{\trt\trs}=\rel{\trt}\rel{\trs}$; $\trt\lor\trs$ denotes a transducer such that  $\rel{\trt\lor\trs}=\rel{\trt}\cup\rel{\trs}$.

\section{Statement and Motivation of the Main Problem}\label{sec:theproblem}
%
In this section, we make precise the problem we are dealing with, and we explain its context as well as the motivation for considering it. 
\pnsi
Let $I$ be an irreflexive relation. An \emdef{asymmetric partition} of $I$ is a partition $\{A,B\}$ of $I$ such that $A,B$ are asymmetric. If $I$ is rational, then a \emdef{rational asymmetric partition} of $I$ is an asymmetric partition $\{A,B\}$ of $I$ such that $A,B$ are rational.   

\begin{remark}
	If $I$ is any irreflexive relation and $[>]$ is any total asymmetry then $\{I\cap[>],\,I\cap[<]\}$ is an asymmetric partition of $I$. As any asymmetric $A$ is irreflexive, we also have that $\{A\cap[>],\,A\cap[<]\}$ is an asymmetric partition of $A$. If $S$ is a symmetric and irreflexive relation and $\{A,B\}$ is an asymmetric partition of $S$ then $B=A^{-1}$.
\end{remark}

\pssn
\textbf{The Rational Asymmetric Partition Problem.} 
Which symmetric-and-irreflexive  rational relations have a rational asymmetric partition?

\begin{remark}
	Any relation $R$ that is not irreflexive cannot have an asymmetric partition; otherwise, $R$ would contain a pair $u/u$ and this cannot be an element of any asymmetric relation. We also have the following observation. 
	\begin{itemize}
	\item
	If $A$ is any rational asymmetric relation then $\{A,A^{-1}\}$ is a rational asymmetric partition of $A\cup A^{-1}$.
	\end{itemize} 
\end{remark}

\pssn
\textbf{Motivation for the above problem.} For a relation $R$ and language $L$, we say that $L$ is $R$-independent, \cite{Shyr:Thierrin:relations,SSYu:book}, if
$$uRv, u\in L, v\in L \to u=v.$$
If $R$ is irreflexive and realized by some transducer \trt then, \cite{Kon:2017}, the above condition is equivalent to
$\trt(L)\cap L=\emptyset.$
In any case, we have that $L$ is $R$-independent, if and only if it is $(R\cup R^{-1})$-independent; and of course the relation $(R\,\cup R^{-1})$ is \emph{always} symmetric. 
The concept of independence provides tools for studying code-related properties such as prefix codes and error-detecting languages (according to the relation $R$).   
In \cite{KonMas:2017}, for a given input-altering transducer \trt and regular language $L$ that is $\rel{\trt}$-independent, the authors provide a formula for embedding $L$ into a maximal $\rel{\trt}$-independent language, provided that \trt is input-decreasing with respect to $[>_r]$. Of course then, $\rel{\trt}$ is asymmetric. Thus, to embed an $S$-independent language $L$ into a maximal one, where $S$ is symmetric, it is necessary to find a transducer \trt such that $S=\rel{\trt}\cup\rel{\trtinv}$ and $\rel\trt$ is asymmetric .

\section{Multicopies of Transducers}\label{sec:multicopies}
%
In this section we fix a finite nonempty set $C$, whose elements are called \emdef{copy labels}. Let $S$ be any set and let $c\in C$. The \emdef{copy} $c$ of $S$ is the set $S^c=\{s^c\mid s\in S\}$.

\begin{definition}\label{def:tr:copy}\NOTE{def:tr:copy}
Let $\trt = (Q,\Sigma, T, I, F)$ be a transducer. A \emdef{$C$-copy} of \trt is any transducer $\trt'=(Q',\Sigma,T',I',F')$ satisfying the following conditions.
\begin{enumerate}
	\item 
	$Q'=\cup_{c\in C}Q^c$, $\>I'\sse \cup_{c\in C}I^c$,  $\>F'\sse\cup_{c\in C}F^c$.
	\item 
	$T'\sse\big\{(p^c,x/y,q^d)\mid c,d\in C,\,(p,x/y,q)\in T\big\}$. If  $e'=(p^c,x/y,q^d)\in T'$ then the  edge $(p,x/y,q)$ of $\trt$ is called \emdef{the edge of \trt corresponding to} $e'$ and is denoted by $\corr{e'}$.
\end{enumerate}
For each edge $e$ of $\trt$,  we define the \emdef{set of edges of $\trt'$ corresponding to} $e$ to be the set
	\[
	\Corr{e}=\{e'\mid e=\corr{e'}\}.
	\]
\end{definition}

\begin{example}
	The transducer $\alpha_0(\trs)$ in Fig.~\ref{fig:example2} is a  $C$-copy of \trs, where $C=\{\ew,A,R\}$. 
	It has three copies of the states of \trs. We have that 
	\[
	\Corr{q_1,0/1,q_2}=\{(q_1^\ew,0/1,q_2^R),\;(q_1^A,0/1,q_2^A),\;(q_1^R,0/1,q_2^R)\}.
	\] 
	Each edge $(p,x/y,q)$ of \trs has corresponding edges in $\alpha_0(\trs)$ of the form $(p^c,x/y,q^d)$ such that the source state $p^c$ is in the copy $c$ (initially, $c=\ew$) and the destination state $q^d$ is in the copy $d$, where possibly $d=c$. Edges of $\alpha_0(\trs)$ with source state in the copies $A,R$ have a destination state  in the same copy. On the other, an edge of $\alpha_0(\trs)$, with some label $x/y$, whose source state is in the copy $\ew$ has  a destination state in the copy $\ew$ if $x=y$; in the copy $A$ if $x>_ry$; and in the copy $R$ if $x<_ry$. As $\alpha_0(\trs)$ has final states only in the copy $A$, it follows that for any $u/v\in\rel{\alpha_0(\trs)}$ we have that $u>_rv$ and $u/v\in\rel{\trs}$.
	This example is useful when solving the rational asymmetric partitioning problem for letter-to-letter transducer---see Section~\ref{sec:letter}. 
\end{example}

\begin{remark}\label{rem:tr:copy}\NOTE{rem:tr:copy}
	The below observations follow from the above definitions and are simple and helpful facts to use when proving statements about $C$-copies of transducers. Let $\trt$ be a transducer and let $\trt'$ be a $C$-copy of \trt. 
	\begin{enumerate}
	    \item The set of edges of $\trt'$ is equal to $\bigcup_{e\in T}\Corr{e}$, where $T$ is the set of edges of \trt. Thus, to define the edges of a $C$-copy of \trt, it is sufficient to specify, for all edges $e$ of \trt, the sets $\Corr e$. 
		\item If $\trt'$ has a state that is both initial and final then so does \trt. Thus, if $\ep\in\aComput{\trt'}$ then we have that $\ep\in\aComput{\trt}$ and $\ewp\in\rel\trt.$
	\end{enumerate}
\end{remark}

\begin{definition}\label{D:corr:path}
Let $\trt'$ be a $C$-copy of a transducer \trt, and let $P'=\walk{e_i'}_{i=1}^{\ell}\in\Path{\trt'}\sm\{\ep\}$. For each edge 
$e_i'$ of $P'$, let $e_i=\corr{e_i'}$. Then the string 
$\walk{e_i}_{i=1}^\ell$ is a path of \trt and is called  the (unique) \emdef{path of \trt corresponding to} $P'$ and is denoted by 
$$\corr{P'}.$$ 
Conversely, if $P=\walk{e_i}_{i=1}^{\ell}$ is a path of \trt, then we define the \emdef{set of paths of $\trt'$ corresponding to} $P$ to be the set of all paths of $\trt'$ of the form $\walk{e_i'}_{i=1}^{\ell}$, where each $e_i'\in\Corr{e_i}$; this set is denoted by 
$$\Corr{P}.$$
    We also define $\corr{\ep}=\ep$ and $\Corr{\ep}=\{\ep\}$.
\end{definition}

\begin{remark}\label{rem:tr:copy2}\NOTE{rem:tr:copy2}
	Let $\trt$ 
	be a transducer and let $\trt'$  
	be a $C$-copy of \trt. Let $P\in\Path{\trt}$ and $P'\in\Path{\trt'}$.
	\begin{enumerate}
	    \item We have that $P'\in\Corr{P}$ if and only if $P=\corr{P'}$.
		\item If $e'$ is an edge of $\trt'$ and $P'e'\in\Path{\trt'}$, then $\corr{P'e'}=\corr{P'}\,\corr{e'}$.
	\end{enumerate}
\end{remark}

\begin{lemma}\label{L:faithful}
	If $\trt'$ is a  $C$-copy of a transducer \trt, then 
	$\rel{\trt'}\sse\rel{\trt}. $
\end{lemma}
\begin{proof}
	Let $u/v\in\rel{\trt'}$. We show that $u/v\in\rel\trt$. There is $P'\in\aComput{\trt'}$ with $\lbl{P'}=u/v$. If $P'=\ep$ then $u/v=\ewp$ and the statement follows from Remark~\ref{rem:tr:copy}. Now suppose that $P'=\walk{q_{i-1}^{c_{i-1}},u_i/v_i,q_{i}^{c_i}}_{i=1}^{\ell}$ with $\ell>0$. Then, $P=\corr{P'}$ is a path of \trt. Moreover, as $q_0^{c_0}$ is initial in $\trt'$ and $q_\ell^{c_\ell}$ is final in $\trt'$, the state $q_0$ is initial in \trt and $q_\ell$ is final in \trt; hence $P\in\aComput{\trt}$. 
	Also as $P= \walk{q_{i-1},u_i/v_i,q_i}_{i=1}^{\ell}$, we have $\lbl P=u/v$. Hence, $u/v\in\rel\trt$.
\end{proof}

\section{Asymmetric Partition of Letter-to-letter Transducers}\label{sec:letter}
A transducer \trt is called \emdef{letter-to-letter}, \cite{Sak:2009}, if all its transition labels are of the form $\sigma/\tau$, where $\sigma,\tau\in\al$. Here we provide a solution to the asymmetric partition problem for letter-to-letter transducers in Proposition~\ref{P:ltl}, which is based on Construction~\ref{con:ltl} below. We note that this construction is a special case of the more general construction for zero-avoiding transducers in Section~\ref{sec:k:noncrossing:partition}, but we present it separately here as it is simpler than the general one.

\begin{construction}\label{con:ltl}\NOTE{con:ltl}
Let $\trs = (Q,\Sigma, T, I, F)$ be a letter-to-letter transducer. Let $C=\{\ew,A,R\}$. We construct a transducer $\alpha_0(\trs)=(Q',\Sigma,T',I',F')$, which is a  $C$-copy of $\trs$, as follows. 
\pmsi
First, $\>Q'=Q^\ew\cup Q^A\cup Q^R$,\quad $I'=I^\ew\>$ and $\>F'=F^A$. 
\pssi Then, $T'$ is defined  as follows.
\begin{eqnarray*}
	T' & =   & \{(p^c,\sigma/\tau,q^c)\mid (p,\sigma/\tau,q)\in T,\>c\in \{A,R\}\} \\
	   &\cup & \{(p^\ew,\sigma/\sigma,q^\ew)\mid (p,\sigma/\sigma,q)\in T\} \\
	   &\cup & \{(p^\ew,\sigma/\tau,q^A)\mid (p,\sigma/\tau,q)\in T,\>\sigma>_r\tau\} \\
	   &\cup & \{(p^\ew,\sigma/\tau,q^R)\mid (p,\sigma/\tau,q)\in T,\>\sigma<_r\tau\}.
\end{eqnarray*}
\end{construction}

\pnsn\textbf{Explanation.}
The constructed transducer $\alpha_0(\trs)$ contains two exact copies of \trs: a copy whose states are the $A$ copies of $Q$, and a copy whose states are the $R$ copies of $Q$; it also contains a sub-copy of \trs which contains a $\ew$ copy of $Q$ and only transitions with labels of the form $\sigma/\sigma$. Any computation $P'$ of $\alpha_0(\trs)$ starts at an initial state $i^\ew$ and continues with states in the $\ew$ copy of $Q$ as long as transition labels are of the form $\sigma/\sigma$. If a transition label is $\sigma/\tau$ with $\sigma>_r\tau$ then the computation $P'$ continues in the $A$ copy and never leaves that copy. As final states are only in the $A$ copy, we have that $P'$ is accepting if and only if $\corr{P'}$ is accepting and $\lbl{P'}=u/v$ such that $u$ is of the form $x\sigma y_1$ and $v$ of the form $x\tau y_2$ with $\sigma>_r\tau$ and $\len u=\len v$. Note that, in the computation $P'$, if a transition label is $\sigma/\tau$ with $\sigma<_r\tau$ and the current state is in the $\ew$ copy, then $P'$ would continue in the $R$ copy of $\alpha_0(\trs)$, which has no final states, so $P'$ would not be accepting.

\begin{remark}
	In fact the $R$ copy of \trs is not necessary as it has no final states. It was included to make the construction a little more intuitive.
\end{remark}

\pmsn Using Lemma~\ref{L:faithful} and based on the above explanation, we have the following lemma---it is a special case of the lemma for zero-avoiding transducers in Section~\ref{sec:k:noncrossing:partition} where a more rigorous proof is given.

\begin{lemma}\label{L:ltl}
	Referring to Construction~\ref{con:ltl}, we have that 
	\[
	\relb{\alpha_0(\trs)}\>=\>\rel\trs\;\cap\;\{u/v:u>_rv\}.
	\]
\end{lemma}

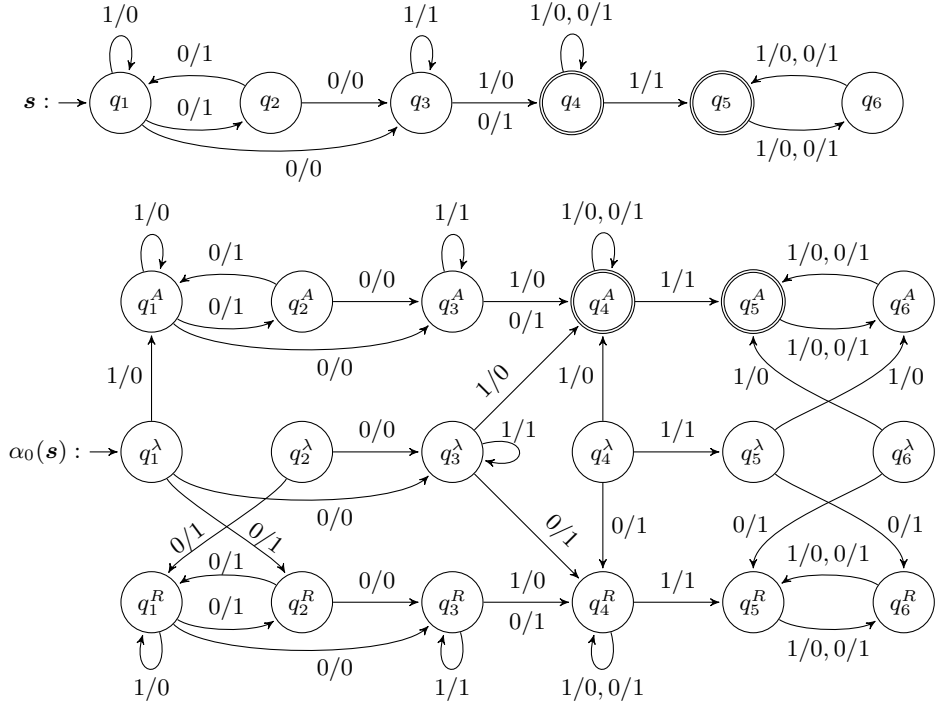
\begin{figure}[hbt]
\begin{center}
\begin{tikzpicture}[>=stealth', initial text={~},shorten >=1pt,auto,node distance=2cm,initial text={$\trs:$}]
\node[state,initial] (1) [label=center: $q_1$] {};
\node[state] [right of=1] (2) [label=center: $q_2$] {};
\node[state] [right of=2] (3) [label=center: $q_3$] {};
\node[state,accepting] [right of=3] (4) [label=center: $q_4$] {};
\node[state,accepting] [right of=4] (5) [label=center: $q_5$] {};
\node[state] [right of=5] (6) [label=center: $q_6$] {};
\draw [->] (1) to [loop above] node {$1/0$} (1);
\draw [to path={
.. controls +(0.6,-0.4) and +(-0.6,-0.4) .. (\tikztotarget) \tikztonodes}]
 [->] (1) to
node[above] {$0/1$} (2);
\drawasxy{1}{3}{$0/0$}{below,pos = 0.6}{-45}{1.5}{225}{1.5}
\draw [to path={
.. controls +(-0.6,0.4) and +(0.6,0.4) .. (\tikztotarget) \tikztonodes}]
 [->] (2) to
node[above] {$0/1$} (1);
\draw [->] (2) to  node{$0/0$} (3);
\draw [->] (3) to [loop above] node {$1/1$} (3);
\draw [->] (3) to  node[above]{$1/0$} node[below]{$0/1$} (4);
\draw [->] (4) to [loop above] node {$1/0,0/1$} (4);
\draw [->] (4) to  node {$1/1$} (5);
\draw [to path={
.. controls +(0.6,-0.4) and +(-0.6,-0.4) .. (\tikztotarget) \tikztonodes}]
 [->] (5) to
node[below] {$1/0,0/1$} (6);
\draw [to path={
.. controls +(-0.6,0.4) and +(0.6,0.4) .. (\tikztotarget) \tikztonodes}]
 [->] (6) to
node[above] {$1/0,0/1$} (5);
\end{tikzpicture}

\begin{tikzpicture}[>=stealth', initial text={~},shorten >=1pt,auto,node distance=2cm,initial text={$\alpha_0(\trs):$}]
\node[state,initial] (1) [label=center: $q_{1}^{\lambda}$] {};
\node[state] [right of=1] (2) [label=center: $q_{2}^{\lambda}$] {};
\node[state] [right of=2] (3) [label=center: $q_{3}^{\lambda}$] {};
\node[state] [right of=3] (4) [label=center: $q_{4}^{\lambda}$] {};
\node[state] [right of=4] (5) [label=center: $q_{5}^{\lambda}$] {};
\node[state] [right of=5] (6) [label=center: $q_{6}^{\lambda}$] {};

\node[state,draw=none] [below of=1] (x) {};

\node[state] [above of=1] (1') [label=center: $q_{1}^{A}$] {};
\node[state] [right of=1'] (2') [label=center: $q_{2}^{A}$] {};
\node[state] [right of=2'] (3') [label=center: $q_{3}^{A}$] {};
\node[state,accepting] [right of=3'] (4') [label=center: $q_{4}^{A}$] {};
\node[state,accepting] [right of=4'] (5') [label=center: $q_{5}^{A}$] {};
\node[state] [right of=5'] (6') [label=center: $q_{6}^{A}$] {};

\node[state] [below of=1] (1D) [label=center: $q_{1}^{R}$] {};
\node[state] [right of=1D] (2D) [label=center: $q_{2}^{R}$] {};
\node[state] [right of=2D] (3D) [label=center: $q_{3}^{R}$] {};
\node[state] [right of=3D] (4D) [label=center: $q_{4}^{R}$] {};
\node[state] [right of=4D] (5D) [label=center: $q_{5}^{R}$] {};
\node[state] [right of=5D] (6D) [label=center: $q_{6}^{R}$] {};
\drawasxy{1}{3}{$0/0$}{below,pos = 0.6}{-45}{1.5}{225}{1.5}
\draw [->] (2) to  node{$0/0$} (3);
\draw [->] (3) to [loop right] node[above] {$1/1$} (3);
\draw [->] (4) to  node {$1/1$} (5);

\draw [->] (1') to [loop above] node {$1/0$} (1');
\draw [to path={
.. controls +(0.6,-0.4) and +(-0.6,-0.4) .. (\tikztotarget) \tikztonodes}]
 [->] (1') to
node[above] {$0/1$} (2');
\drawasxy{1'}{3'}{$0/0$}{below,pos = 0.6}{-45}{1.5}{225}{1.5}
\draw [to path={
.. controls +(-0.6,0.4) and +(0.6,0.4) .. (\tikztotarget) \tikztonodes}]
 [->] (2') to
node[above] {$0/1$} (1');
\draw [->] (2') to  node{$0/0$} (3');
\draw [->] (3') to [loop above] node {$1/1$} (3');
\draw [->] (3') to  node[above]{$1/0$} node[below]{$0/1$} (4');
\draw [->] (4') to [loop above] node {$1/0,0/1$} (4');
\draw [->] (4') to  node {$1/1$} (5');
\draw [to path={
.. controls +(0.6,-0.4) and +(-0.6,-0.4) .. (\tikztotarget) \tikztonodes}]
 [->] (5') to
node[below] {$1/0,0/1$} (6');
\draw [to path={
.. controls +(-0.6,0.4) and +(0.6,0.4) .. (\tikztotarget) \tikztonodes}]
 [->] (6') to
node[above] {$1/0,0/1$} (5');

\draw [->] (1) to  node{$1/0$} (1');
\draw [->] (3) to  node[sloped]{$1/0$} (4');
\draw [->] (4) to  node{$1/0$} (4');
\draw [->] (4) to  node[right]{$0/1$} (4D);

\def\vectora{1}
\drawasxy{5}{6'}{$1/0$}{right,pos = 0.7,yshift=-3pt}{45}{\vectora}{-90}{\vectora}
\drawasxy{6}{5'}{$1/0$}{left,pos = 0.7,yshift=-3pt}{135}{\vectora}{-90}{\vectora}

\draw [->] (1D) to [loop below] node {$1/0$} (1D);
\draw [to path={
.. controls +(0.6,-0.4) and +(-0.6,-0.4) .. (\tikztotarget) \tikztonodes}]
 [->] (1D) to
node[above] {$0/1$} (2D);
\drawasxy{1D}{3D}{$0/0$}{below,pos = 0.6}{-45}{1.5}{225}{1.5}
\draw [to path={
.. controls +(-0.6,0.4) and +(0.6,0.4) .. (\tikztotarget) \tikztonodes}]
 [->] (2D) to
node[above,yshift=-3pt] {$0/1$} (1D);
\draw [->] (2D) to  node{$0/0$} (3D);
\draw [->] (3D) to [loop below] node {$1/1$} (3D);
\draw [->] (3D) to  node[above]{$1/0$} node[below]{$0/1$} (4D);
\draw [->] (4D) to [loop below] node {$1/0,0/1$} (4D);
\draw [->] (4D) to  node {$1/1$} (5D);
\draw [to path={
.. controls +(0.6,-0.4) and +(-0.6,-0.4) .. (\tikztotarget) \tikztonodes}]
 [->] (5D) to
node[below] {$1/0,0/1$} (6D);
\draw [to path={
.. controls +(-0.6,0.4) and +(0.6,0.4) .. (\tikztotarget) \tikztonodes}]
 [->] (6D) to
node[above] {$1/0,0/1$} (5D);

\drawasxy{1}{2D}{$0/1$}{above,pos = 0.7,yshift=-3pt,sloped}{-60}{1}{120}{1}
\drawasxy{2}{1D}{$0/1$}{above,pos = 0.7,yshift=-3pt,sloped}{-120}{1}{60}{1}
\draw [->] (3) to  node[sloped]{$0/1$} (4D);
\def\vectora{1}
\drawasxy{5}{6D}{$0/1$}{right,pos = 0.7,yshift=3pt}{-45}{\vectora}{90}{\vectora}
\drawasxy{6}{5D}{$0/1$}{left,pos = 0.7,yshift=3pt}{-135}{\vectora}{90}{\vectora}
\end{tikzpicture}
\pmsn
\parbox{0.9\textwidth}{\caption{Example of Construction~\ref{con:ltl} applied to transducer $\trs$ to get transducer $\alpha_0(\trs)$.}\label{fig:example2}
}
\end{center}
\end{figure}

\begin{proposition}\label{P:ltl}\NOTE{P:ltl}
Let $\trs$ be any input-altering letter-to-letter  transducer. Let $\trt_1=\alpha_0(\trs)$ and $\trt_2=\big(\alpha_0(\trsinv)\big)^{-1}$, where $\alpha_0(\trs)$ is the transducer produced in Construction~\ref{con:ltl}. The following statements hold true.
\begin{enumerate}
\item $|\alpha_0(\trs)|=\Theta(|\trs|)$.
\item $\{\rel{\trt_1},\rel{\trt_2}\}$ is a rational asymmetric partition of $\rel{\trs}$.
\end{enumerate} 
\end{proposition}

\begin{proof}
It follows from the construction and the lemma.
\end{proof}

\section{Discrepancies of Computations and  Zero-avoiding Transducers}\label{sec:k:noncrossing}
Here we introduce the concept of a zero-avoiding transducer with some bound $k\in\N_0$, which relates to length discrepancies of the computations of the transducer. We show that the bound $k$ is always less than the number of states of the transducer. We also show that the class of relations realized by zero-avoiding transducers (of any bound) is a proper superset of all left and right synchronous relations. Thus, they also include all recognizable relations and all relations of bounded length discrepancy.

\begin{definition}\label{D:discrep}\NOTE{D:discrep}
Let $u,v\in\alstar$, let \trt be a transducer and let $P=\walk{q_{i-1},x_i/y_i,q_i}_{i=1}^{\ell}\in\Path{\trt}$.
\begin{itemize}
\item[(i)] The \emdef{length discrepancy} of the pair $u/v$ is
the integer $d(u/v)=|u|-|v|$.
\item[(ii)] The \emdef{length discrepancy} of $P$ is
the integer 
$$d(P)= d(x_1x_2\cdots x_\ell / y_1y_2\cdots y_\ell).$$
\item[(iii)] The \emdef{maximum absolute length discrepancy} of $P$ is the integer
$$d_{max}(P)= \max_{Q\in\Px{P}}\left\{| d(Q)| \right\}.$$
\end{itemize}
\end{definition}

\begin{remark}
	We have that $d(\ewp)=0$ and $d_{max}(\ep)=0$. 
	Moreover, 
	\pssi $\blacktriangleright$ if $P_1P_2$ is a path of \trt, then $d(P_1P_2)=d(P_1)+d(P_2)$.
\end{remark}

\begin{definition}\label{D:0avoid}\NOTE{D:0avoid} 
A transducer \trt is called \emdef{zero-avoiding}, if there is an integer $k\ge0$ such that the following condition is satisfied: 
\begin{center}
for any  $P\in\Comput{\trt}$,\;  if $d_{max}(P)>k$ then $d(P)\neq0$. 
\end{center}
In this case, \trt is called zero-avoiding \emdef{with bound $k$}. It is called zero-avoiding with \emdef{minimum bound} $k$, if it is zero-avoiding with bound $k$ and not zero-avoiding with bound $k-1$.\footnote{This is well-defined: if \trt is zero-avoiding with bound $k$ then it is also zero-avoiding with bound $k'$ for all $k'>k$. }
\end{definition}

\begin{remark}
In the above definition of a zero-avoiding transducer, if a computation $P$ has length discrepancy $>k$, or $<-k$, then any continuation of $P$ cannot have zero as its length discrepancy.
\end{remark}

\begin{remark}
Let $\trt = (Q,\Sigma, T, I, F)$ be a transducer. For any path $P$ of \trt there is a unique path $P^{-1}$ of \trtinv whose labels are the inverses of the labels in $P$. Thus, $d(P^{-1})=-d(P)$ and $|d(P^{-1})|=|d(P)|$. This implies that
\pnsi $\blacktriangleright$ if \trt is zero-avoiding with some bound $k$ then also \trtinv is zero-avoiding with  bound $k$.
\end{remark}

\begin{remark}\label{rem:copy:discr}\NOTE{rem:copy:discr}
	Let $\trs = (Q,\Sigma, T, I, F)$ be a transducer and let $\trt=(Q',\Sigma,T',I',F')$ be a  $C$-copy of \trs. Let $P'$ be a computation of \trs, and let $P=\corr{P'}$; then $P$ and $P'$ have exactly the same sequence of labels in their transitions. Thus the following statements hold.
	\begin{enumerate}
	    \item $d_{max}(P')=d_{max}(P)$.
	    \item $d(Q)=d(Q')$ for any prefixes $Q,Q'$ of $P,P'$ of the same length. 
		\item If $\trs$ is zero-avoiding with some bound $k$ then also \trt is zero-avoiding with the same bound $k$.
	\end{enumerate}
\end{remark}

\begin{lemma}\label{L:signed:cycles}\NOTE{L:signed:cycles}
	Let \trt be an $n$-state transducer, for some integer $n>0$.
	\begin{enumerate}
		\item If \trt has a path $P$ with $d(P)\ge n$ (resp., $d(P)\le-n$) then $P$ contains a cycle $B$ with $d(B)>0$ (resp., $d(B)<0$).
        \item If \trt has a computation $P$ with $d_{max}(P)\ge n$ and $d(P)=0$, then $P=BC_1AC_2D$ such that  $C_1,C_2$ are cycles with $d(C_1)d(C_2)<0$.
	\end{enumerate}
\end{lemma}
\begin{proof} 
The proofs of the  statements make use of the path terminology in Section~\ref{sec:notation}. For the first statement, we show only the case $d(P)\ge n$, as the other case is symmetric. As every transition in $P$ changes the discrepancy by at most one, $P$ contains at least $n$ transitions, so there are at least $n+1$ states in $P$, which implies that $P$ contains a cycle. Let $\cS$ be the set of cycles $C$ of $P$ with $d(C)\le0$. Then, $d(P-\cS)\ge d(P)$ and, as $d(P-\cS)\ge n$, $P-\cS$ must contain a cycle $B$, which must have $d(B)>0$. Then, the cycle $B$ is also contained in $P$.
\par 
For the second statement, we split the computation $P$ into two paths $P=P_1P_2$ such that $P_1$ is the shortest prefix of $P$ with $|d(P_1)|=n$. First we consider the case where $d(P_1)=n$. Then, $P_1$ contains a cycle $C_1$ with $d(C_1)>0$. 
Then, $d(P_2)=-n$, which implies that $P_2$ contains a cycle $C_2$ with $d(C_2)<0$. The case where $d(P_1)=-n$ is analogous.
\end{proof}

\begin{proposition}\label{P:mink}\NOTE{P:mink}
Let \trt be an $n$-state transducer, for some integer $n\ge1$. 
If $\trt$ is  zero-avoiding with minimum bound $k$ then $k<n$.
\end{proposition}
\begin{proof}
The  statement holds trivially if $k=0$. So suppose $k\ge1$ and assume  for the sake of contradiction that $k\geq n$. As \trt is not zero-avoiding with bound $k-1$, there is a computation $P$ of \trt such that $d_{max}(P)=k$ and $d(P)=0$. Then, Lemma~\ref{L:signed:cycles} implies that $P=BC_1AC_2D$ for some cycles $C_1,C_2$ with $d(C_1)d(C_2)<0$. We can use the cycles $C_1,C_2$ to make a new computation $Q$ of \trt such that $d_{max}(Q)>k$ and $d(Q)=0$. Without loss of generality, assume that $d(C_1)>0$ and $d(C_2)<0$. The required computation of \trt is
    \[
    Q=BC_1\,C_1^{|d(C_2)|k\big(1+|d(B)|\big)}\,AC_2\,C_2^{d(C_1)k\big(1+|d(B)|\big)}\,D.
    \] 
    Then we have that $d(Q)=d(P)+d(C_1)|d(C_2)|k\big(1+|d(B)|\big)+d(C_2)d(C_1)k\big(1+|d(B)|\big)=0$ and
    \begin{align*}
    d_{max}(Q) &\ge d_{max}\Big(BC_1C_1^{|d(C_2)|k\big(1+|d(B)|\big)}\Big) \\ 
       &\ge d(B)+d(C_1)+d(C_1)|d(C_2)|k\big(1+|d(B)|\big)>k
    \end{align*}
    which contradicts the assumption that \trt is zero-avoiding with bound $k$. 
\end{proof}

\begin{proposition}\label{P:zavoidIFF}
Let \trt be an $n$-state transducer, for some integer $n\ge1$. The following statements are equivalent.
\begin{enumerate}
	\item \trt is not zero-avoiding
	\item \trt has a computation $P$ with $d_{max}(P)\ge n$ and $d(P)=0$.
    \item \trt has a computation $P$ of the form $P=BC_1AC_2D$ such that  $C_1,C_2$ are cycles with $d(C_1)d(C_2)<0$.
\end{enumerate}
\end{proposition}
\begin{proof}
	That the first statement is equivalent to the second one follows logically  from Definition~\ref{D:0avoid} and the above lemma. Now we show that the second and third statements are equivalent. The second statement implies the third one, by Lemma~\ref{L:signed:cycles}. Now suppose that the third statement holds, that is, there is a computation $P=BC_1AC_2D$ of \trt such that $d(C_1)d(C_2)<0$. As in the proof of the above lemma, we can use the cycles $C_1,C_2$ to make a computation $Q$ of \trt with $d_{max}(Q)\ge n$ and $d(Q)=0$.
\end{proof}

\if\DRAFT1\pssn\sk{QUESTION: }{The above proposition should allow us to decide whether a given transducer \trt is zero-avoiding. Is that correct? Can we do it in polynomial time? I think Juraj had looked at this question.}
\fi

\pssn\textbf{Relating left (right) synchronous and zero-avoiding relations.}
A natural question that arises is how zero-avoiding relations are related to the well-known left (or right) synchronous relations. A relation $R$ is called \emdef{left synchronous} if the relation $\overrightarrow R$ can be realized by a letter-to-letter transducer over the alphabet $\Sigma\cup\{\#\}$ with $\#\notin\Sigma$ \cite{Cart:09,Chof:2006}. Here we use the notation 
$$\overrightarrow{u/v}=(u/v\#^{|u|-|v|}) \text{ if } |u|\ge|v|; \text{ and } \overrightarrow{u/v}=(u\#^{|v|-|u|}/v) \text{ if } |u|<|v|.$$
Then, $\overrightarrow R=\{\overrightarrow{u/v}:u/v\in R\}$. An equivalent definition is given in \cite{Sak:2009}: $R$ is left synchronous if it is a finite union of relations, each of the form $S(A\times\{\ew\})$ or $S(\{\ew\}\times A)$, where $A$ is a regular language and $S$ is realized by a letter-to-letter transducer. The concept of a right synchronous relation is symmetric: we can define it either via 
$$\overleftarrow{u/v}=(u/\#^{|u|-|v|}v) \text{ if } |u|\ge|v|; \text{ and } \overleftarrow{u/v}=(\#^{|v|-|u|}u/v) \text{ if } |u|<|v|,$$
or via finite unions of relations of the form $(A\times\{\ew\})S$ or $(\{\ew\}\times A)S$.

\begin{proposition}\label{P:lsyncVSzavoid}\NOTE{P:lsyncVSzavoid}
	The classes of left synchronous relations and right synchronous relations are proper subsets of the class of zero-avoiding relations with bound 0.
\end{proposition}
\begin{proof}
Consider any left synchronous relation $R=S_1C_1\cup\cdots\cup S_mC_m$ such that each $S_i$ is realized by a letter-to-letter transducer $\trt_i$ and each $C_i$ is of the form $(A_i\times\{\ew\})$ or $(\{\ew\}\times A_i)$, where $A_i$ is a regular language. Then, each $C_i$ is realized by a transducer $\trs_i$ whose transition labels are either all in $\Sigma\times\{\ew\}$, or all in $\{\ew\}\times\Sigma$. The zero-avoiding transducer realizing $R$ is the transducer $\trt=\bigvee_{i=1}^m\trt_i\trs_i$. It is zero-avoiding with bound 0 because each computation $P$ of \trt is in one component $\trt_i\trs_i$ and, once a prefix $Q$ of $P$ has $d(Q)>0$ (resp., $d(Q)<0$) then also $d(P)>0$ (resp., $d(P)<0$). Similarly, one has that every right synchronous relation is also zero-avoiding with bound 0.
\pnsi
The rational relation $R=(00/0)^*$ is zero-avoiding but not left synchronous. 
To see that $R=(00/0)^*$ is zero-avoiding we use the transducer in Fig.~\ref{fig:aa,a} realizing $R$ and satisfying the zero-avoiding condition with bound 0.
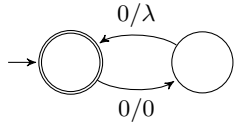
\begin{figure}[hbt]
\begin{center}
\begin{tikzpicture}[>=stealth', initial text={~},shorten >=1pt,auto,node distance=1.75cm,initial text={}]
\node[state,initial,accepting] (2) {};
\node[state] [right of=2] (3) {};
\draw [to path={
.. controls +(0.6,-0.4) and +(-0.6,-0.4) .. (\tikztotarget) \tikztonodes}]
 [->] (2) to
node[below] {$0/0$} (3);
\draw [to path={
.. controls +(-0.6,0.4) and +(0.6,0.4) .. (\tikztotarget) \tikztonodes}]
 [->] (3) to
node[above] {$0/\ew$} (2);
\end{tikzpicture}
\end{center}
\caption{Transducer realizing the relation $(00/0)^*$.}
\label{fig:aa,a}
\end{figure}
%
To show that $R$ is not left synchronous, we use contradiction. Assume that it is, then the relation $\overrightarrow R=\{(0^i0^i/0^i\#^i):i\in\N_0\}$ would be rational, which implies that the language $\{0^i\#^i:i\in\N_0\}$ would be regular; a contradiction. Similarly, we have that $R$ is not right synchronous.
\end{proof}

\pnsn\textbf{Zero-avoiding relations are not closed under intersection.} To see this consider the zero-avoiding relations
\[
(0/0)^*(\ew/1)^*\quad\text{ and }\quad (0/\ew)^*(1/0)^*.
\]
Their intersection is $\{0^i1^i/0^i:i\in\N_0\}$, which is non-rational (if it were rational then the language $\{0^i1^i:i\in\N_0\}$ would be regular).

\section{Asymmetric Partition of Zero-avoiding Transducers}\label{sec:k:noncrossing:partition}
We present next a solution to the asymmetric partition for any relation realized  by some zero-avoiding transducer $\trs$---see Construction~\ref{con:zero} and Theorem~\ref{th:qs}. The required asymmetric relation is realized by a transducer $\alpha(\trs)$ which is a $C$-copy of $\trs$, where $C$ is shown in~\eqref{eq:copies} further below. In fact $\rel{\alpha(\trs)}=(\rel\trs\cap[>_r])$; thus, $u/v\in\rel{\alpha(\trs)}$ implies $u>_r v$. The set of states of $\alpha(\trs)$ is $Q'=\cup_{c\in C}Q^c$. The reason why all these copies of $Q$ are needed is to know at any point during a computation $P'$ of $\alpha(\trs)$ whether $d_{max}(P')$ has exceeded $k$, where $k$ is either the known bound of zero-avoidance, or the number of states of $\trs$. 
\pmsn
\textbf{Meaning of states of $\alpha(\trs)$ in Construction~\ref{con:zero}.} A state $q^c$ of $\alpha(\trs)$ has the following meaning. Let $P'$ be any computation of $\alpha(\trs)$ ending with $q^c$ and having some label $\win/\wout$. Then, state $q^c$ specifies which one of the following mutually exclusive facts about $P'$ holds. 
\begin{itemize}
\item $q^c=q^\ew$ means: $\win=\wout$.
\item  $q^c=q^{+u}$ means: $\win=\wout u$, for some word $u$ with $1\le|u|\le k$, so $\win>_r\wout$. 
\item $q^c=q^{-u}$ means: $\wout=\win u$, for some word $u$ with $1\le|u|\le k$, so $\win<_r\wout$. 
\item $q^c=q^{A\ell}$ means: $\win=x\sigma y$, $\wout= x\tau z$, $\sigma>_r\tau$,  $\ell=|y|-|z|=d(P')$, and $-k\le\ell\le k$. Note that the $A$ in $q^{A\ell}$ is a reminder of $\sigma>_r\tau$ and indicates that $P'$ could be the prefix of an \underline{A}ccepting computation $Q'$ having $d(Q')\ge0$, in which case $\win'>_r\wout'$ where $\win'/\wout'$ is the label of $Q'$.
\item $q^c=q^{R\ell}$ means: $\win=x\sigma y$, $\wout =x\tau z$, $\sigma<_r\tau$, $\ell=|y|-|z|=d(P')$, and $-k\le\ell\le k$. Note that the $R$ in $q^{R\ell}$ is a reminder of $\sigma<_r\tau$ and indicates that $P'$ could be the prefix of a \underline{R}ejecting computation $Q'$ having $d(Q')\le0$, in which case $\win'<_r\wout'$ where $\win'/\wout'$ is the label of $Q'$.
\item $q^c=q^{A}$ means: $d_{max}(P')>k$ and $d(P')=|\win|-|\wout|>0$. 
\item $q^c=q^{R}$ means: $d_{max}(P')>k$ and $d(P')=|\win|-|\wout|<0$.
\end{itemize}
\textbf{Final states in Construction~\ref{con:zero}.}
Based on the meaning of the states and the requirement that the label $\win/\wout$ of an accepting computation $P'$ of $\alpha(\trs)$ satisfies $\win>_r\wout$, the final states of $\alpha(\trs)$ are shown in~\eqref{eq:finals} further below. Let $f$ be any final state of $\trs$. State $f^A$ of $\alpha(\trs)$ is final because, if $P'$ ends in $f^A$, we have $d(P')>0$, which implies $\win>_r\wout$. On the other hand, state $f^R$ is not final because any computation $P'$ of $\alpha(\trs)$ ending in $f^A$  has $d(P')<0$, which implies $\win<_r\wout$. State $f^{R\ell}$, with $\ell>0$, is final because any computation $P'$ of $\alpha(\trs)$ ending in $f^{R\ell}$  has $|\win|-|\wout|=\ell>0$, so $\win>_r\wout$. On the other hand, state $f^{R\ell}$, with $\ell\le0$, is not final because any computation $P'$ of $\alpha(\trs)$ ending in $f^{R\ell}$  has $|\win|-|\wout|=\ell\le0$, and $\win<_r\wout$.

\begin{example}
The transducer $\alpha(\trs)$  consists of several modified copies of $\trs$ (see Fig.~\ref{fig:example1a}) such that, for any computation $P$ of $\trs$ with label $\win/\wout$ there is at least one corresponding computation of $\alpha(\trs)$ with the same label $\win/\wout$ which goes through copies of  the same states appearing in $P$. The initial states of $\alpha(\trs)$ are in the copy $Q^\lambda$, where any computation involving only states in $Q^\lambda$ has equal input and output labels. 
For a transition $e=(p,1/\ew,q)$ of $\trs$, we have that  $e'=(p^{\ew},1/\ew,q^{+1})$ is one of the transitions of $\alpha(\trs)$ corresponding to $e$, where the transition $e'$ starts at the copy $Q^\ew$ of $Q$ and goes to the copy $Q^{+1}$ of $Q$ (see Fig.~\ref{fig:example1b}). In a computation of $\alpha(\trs)$ that ends in the copy $Q^{+1}$, the input label is of the form $x1$ and the output label is of the form $x$. Then, Fig.~\ref{fig:example1c} shows all possible transitions from state $p^{+1}$ to other states of $\alpha(\trs)$, which could be in the same or different copies of $Q$. 	\qed
\end{example}

\begin{figure}[t!]
\centering
\begin{subfigure}[b]{0.5\textwidth}
\centering
\begin{tikzpicture}[>=stealth', initial text={~},shorten >=1pt,auto,node distance=26pt,initial text={~}, scale=0.75, every node/.style={scale=0.75}]
\def\Xext{32pt}
\def\Yext{12pt}
\node[state,draw=none,inner sep=7pt,minimum size=0pt] (Q+11a) [label=center: $Q^{+11}$] {};
\node[state,draw=none,inner sep=7pt,minimum size=0pt] [right of=Q+11a] (Q+11b) [label=center: $\circledcirc$] {};
\draw[rounded corners=0pt] ($(Q+11a)+(-\Xext,-\Yext)$)
-- ($(Q+11b)+(\Xext,-\Yext)$)
-- ($(Q+11b)+(\Xext,\Yext)$)
-- ($(Q+11a)+(-\Xext,\Yext)$)
-- cycle;

\drawJurBlock{Q+11a}{Q+10a}{$Q^{+10}$}{Q+10b}{$\circledcirc$}{\Xext}{\Yext}
\drawJurBlock{Q+10a}{Q+01a}{$Q^{+01}$}{Q+01b}{$\circledcirc$}{\Xext}{\Yext}
\drawJurBlock{Q+01a}{Q+00a}{$Q^{+00}$}{Q+00b}{$\circledcirc$}{\Xext}{\Yext}
\drawJurBlock{Q+00a}{Q+1a}{$Q^{+1}$}{Q+1b}{$\circledcirc$}{\Xext}{\Yext}
\drawJurBlock{Q+1a}{Q+0a}{$Q^{+0}$}{Q+0b}{$\circledcirc$}{\Xext}{\Yext}

\drawJurBlock{Q+0a}{Qla}{$Q^{\lambda}$}{Qlb}{}{\Xext}{\Yext}

\node[state,draw=none,inner sep=7pt,minimum size=0pt,initial] [left of=Qlb][xshift=0pt] (Xinit) [label=center: ] {};

\drawJurBlock{Qla}{Q-0a}{$Q^{-0}$}{Q-0b}{}{\Xext}{\Yext}
\drawJurBlock{Q-0a}{Q-1a}{$Q^{-1}$}{Q-1b}{}{\Xext}{\Yext}
\drawJurBlock{Q-1a}{Q-00a}{$Q^{-00}$}{Q-00b}{}{\Xext}{\Yext}
\drawJurBlock{Q-00a}{Q-01a}{$Q^{-01}$}{Q-01b}{}{\Xext}{\Yext}
\drawJurBlock{Q-01a}{Q-10a}{$Q^{-10}$}{Q-10b}{}{\Xext}{\Yext}
\drawJurBlock{Q-10a}{Q-11a}{$Q^{-11}$}{Q-11b}{}{\Xext}{\Yext}

\node[state,draw=none,inner sep=7pt,minimum size=0pt] [right of=Q+11a, xshift=110pt] (QA2a) [label=center: $Q^{A2}$] {} ;
\node[state,draw=none,inner sep=7pt,minimum size=0pt] [right of=QA2a] (QA2b) [label=center: $\circledcirc$] {};
\draw ($(QA2a)+(-\Xext,-\Yext)$)
-- ($(QA2b)+(\Xext,-\Yext)$)
-- ($(QA2b)+(\Xext,\Yext)$)
-- ($(QA2a)+(-\Xext,\Yext)$)
-- cycle;

\drawJurBlock{QA2a}{QA1a}{$Q^{A1}$}{QA1b}{$\circledcirc$}{\Xext}{\Yext}
\drawJurBlock{QA1a}{QA0a}{$Q^{A0}$}{QA0b}{$\circledcirc$}{\Xext}{\Yext}
\drawJurBlock{QA0a}{QA-1a}{$Q^{A-1}$}{QA-1b}{$ $}{\Xext}{\Yext}
\drawJurBlock{QA-1a}{QA-2a}{$Q^{A-2}$}{QA-2b}{$ $}{\Xext}{\Yext}
\drawJurBlock{QA-2a}{QR2a}{$Q^{R2}$}{QR2b}{$\circledcirc$}{\Xext}{\Yext}
\drawJurBlock{QR2a}{QR1a}{$Q^{R1}$}{QR1b}{$\circledcirc$}{\Xext}{\Yext}
\drawJurBlock{QR1a}{QR0a}{$Q^{R0}$}{QR0b}{}{\Xext}{\Yext}
\drawJurBlock{QR0a}{QR-1a}{$Q^{R-1}$}{QR-1b}{}{\Xext}{\Yext}
\drawJurBlock{QR-1a}{QR-2a}{$Q^{R-2}$}{QR-2b}{}{\Xext}{\Yext}

\node[state,draw=none,inner sep=7pt,minimum size=0pt] [below of=QR-2a, yshift=-26pt] (QAa) [label=center: $Q^{A}$] {} ;
\node[state,draw=none,inner sep=7pt,minimum size=0pt] [right of=QAa] (QAb) [label=center: $\circledcirc$] {};
\draw ($(QAa)+(-\Xext,-\Yext)$)
-- ($(QAb)+(\Xext,-\Yext)$)
-- ($(QAb)+(\Xext,\Yext)$)
-- ($(QAa)+(-\Xext,\Yext)$)
-- cycle;

\drawJurBlock{QAa}{QRa}{$Q^{R}$}{QRb}{}{\Xext}{\Yext}

\draw [densely dashed] ($(QRa)+(-\Xext -23,-\Yext)$)
-- ($(QA2a)+(-\Xext -23,\Yext +1)$);

\draw [densely dashed] ($(QAa)+(-\Xext -23,\Yext +13)$)
-- ($(QAb)+(\Xext,\Yext +13)$);

\draw [->][yshift=0pt] (Q-11b) to node[below,pos=0.37] [yshift=2pt]{$\lambda/\tau$} (QRa);
\draw [->][yshift=0pt] (Q-10b) to node[below,pos=0.37] [yshift=2pt]{$\lambda/\tau$} (QRa);
\draw [->][yshift=0pt] (Q-01b) to node[below,pos=0.37] [yshift=2pt]{$\lambda/\tau$} (QRa);
\draw [->][yshift=0pt] (Q-00b) to node[below,pos=0.37] [yshift=2pt]{$\lambda/\tau$} (QRa);
\draw [->][yshift=0] (Qlb) to node[pos=0.25] [yshift=-2pt]{$0/1$} (QR0a);
\drawasxy{Q-0b}{Q-1b}{$0/1$}{below,pos = 0.85}{0}{1}{0}{1}
\drawasxy{Qlb}{Q+1b}{$1/\lambda$}{left,pos = 0.5}{15}{1.5}{-15}{1.5}
\drawasxy{Q+00b}{Q+00b}{$0/0$}{below,pos = 0.5,yshift=-3pt,xshift=-5pt}{+30}{1}{0}{1}
\drawasxy{QR-1b}{QR-1b}{$\sigma/\tau$}{below,pos = 0.5,yshift=-3pt,xshift=-5pt}{+30}{1}{0}{1}
\draw [->][yshift=0pt] (Q+11b) to node[below,pos=0.37] [yshift=2pt]{$\sigma/0$} (QA2a);
\drawasxy{QA1b}{QA2b}{$\sigma/\lambda$}{below,pos = 0.15}{0}{1}{0}{1}
\drawasxy{QA-1b}{QA-2b}{$\lambda/\tau$}{above,pos = 0.15,yshift=-1pt}{0}{1}{0}{1}
\drawasxy{QR2b}{QAb}{$\sigma/\lambda$}{left,pos = 0.8}{-45}{2.6}{45}{2.6}
\draw [->][yshift=0] (Q+11b) to node[pos=0.28] [yshift=-10pt]{$\sigma/\lambda$} (QAa);
\end{tikzpicture}

\caption{Overview of copies of states of transducer $\trs$ in transducer $\alpha(\trs)$. The figure also depicts some chosen transitions involving $\sigma,\tau\in \Sigma$. The symbol $\circledcirc$ represents that the copy contains some final states;  initial states are only in copy $Q^\lambda$.}
\label{fig:example1a}

\end{subfigure}
~
\begin{subfigure}[b]{0.4\textwidth}
\centering
\begin{tikzpicture}[>=stealth', initial text={~},shorten
>=1pt,auto,node distance=26pt,initial text={~}, scale=0.75, every node/.style={scale=0.75}]
\def\Xext{32pt}
\def\Yext{13pt}
\node[state,inner sep=7pt,minimum size=22pt] (qil) [label=center: $p^{\lambda}$] {};
\node[state,draw=none,inner sep=7pt,minimum size=0pt] [right of=qil] (qjl) [label=center: $ $] {};
\draw[rounded corners=0pt] ($(qil)+(-\Xext,-\Yext)$)
-- ($(qjl)+(\Xext,-\Yext)$)
-- ($(qjl)+(\Xext,\Yext)$)
-- ($(qil)+(-\Xext,\Yext)$)
-- cycle;

\node[state,draw=none,inner sep=7pt,minimum size=0pt] [above of=qil,yshift=20pt] (qi+1) [label=center: $ $] {};
\node[state,inner sep=7pt,minimum size=22pt] [right of=qi+1] (qj+1) [label=center: $q^{{+}1}$] {};
\draw[rounded corners=0pt] ($(qi+1)+(-\Xext,-\Yext)$)
-- ($(qj+1)+(\Xext,-\Yext)$)
-- ($(qj+1)+(\Xext,\Yext)$)
-- ($(qi+1)+(-\Xext,\Yext)$)
-- cycle;

\node[state,minimum size=28pt,draw=none] [left of=qi+1,xshift=-20pt] (qx+1) [label=center: $Q^{+1}$] {};

\node[state,minimum size=28pt,draw=none] [left of=qil,xshift=-20pt] (qxl) [label=center: $Q^{\lambda}$] {};

\draw [->] (qil) to node[pos=0.2,xshift=5pt]{$1/\lambda$} (qj+1);

\end{tikzpicture}
\caption{If the transducer $\trs$ contains the transition $(p,1/\lambda,q)$ then the transducer $\alpha(\trs)$ contains the transition $(p^{\lambda},1/\lambda,q^{+1})$.}
\label{fig:example1b}

\begin{tikzpicture}[>=stealth', initial text={~},shorten >=1pt,auto,node distance=42pt,initial text={~},scale=0.75, every node/.style={scale=0.75}]
\node[state,minimum size=28pt] (1) [label=center: $p^{+1}$] {};
\node[state,minimum size=28pt] [right of=1] (2) [label=center: $q^{+1}$] {};
\node[state,draw=none,minimum size=28pt] [above of=1] (x1) [label=center: ] {};
\node[state,minimum size=28pt] [above of=x1] (3) [label=center: $q^{+10}$] {};
\node[state,minimum size=28pt] [above of=3] (4) [label=center: $q^{+11}$] {};
\node[state,minimum size=28pt] [below of=1] (5) [label=center: $q^{+0}$] {};
\node[state,minimum size=28pt] [below of=5] (6) [label=center: $q^{\lambda}$] {};
\node[state,minimum size=28pt] [above of =2,right of=2] (8) [label=center: $q^{A0}$] {};
\node[state,minimum size=28pt] [above of=8] (7) [label=center: $q^{A1}$] {};
\draw [->] (1) to node{$1/1$} (2);
\draw [->] (1) to node[xshift=2pt,yshift=-2pt]{$0/\lambda$} (3);
\drawasxy{1}{4}{$1/\lambda$}{left,pos = 0.79}{150}{1.5}{210}{1.5}
\draw [->] (1) to node[xshift=1pt,yshift=0pt,right]{$0/1$} (5);
\drawasxy{1}{6}{$\lambda/1$}{right,pos = 0.77}{210}{1.5}{150}{1.5}
\drawasxy{1}{7}{$1/0,0/0$}{above,pos = 0.90,xshift=-15pt}{75}{1.5}{180}{1.5}
\drawasxy{1}{8}{$\lambda/0$}{above,pos = 0.75}{60}{1.5}{180}{1.5}
\end{tikzpicture}
\caption{Sketch of transitions in transducer $\alpha(\trs)$ with different labels, assuming that there are transitions with every possible label from state $p$ to state $q$ of the transducer $\trs$.}
\label{fig:example1c}
\end{subfigure}
\pmsn
\parbox{0.9\textwidth}{\caption{Sketch of examples of a transducer $\alpha(\trs)$ which is the result of Construction~\ref{con:zero} on some transducer $\trs$. The examples use $\al=\{0,1\}$ and $k=2$. Notice that $0$ in $Q^{+0}$ is the string $0$ and $0$ in  $Q^{A0}$ is the number zero (length discrepancy).}\label{fig:example1}
}
\end{figure}
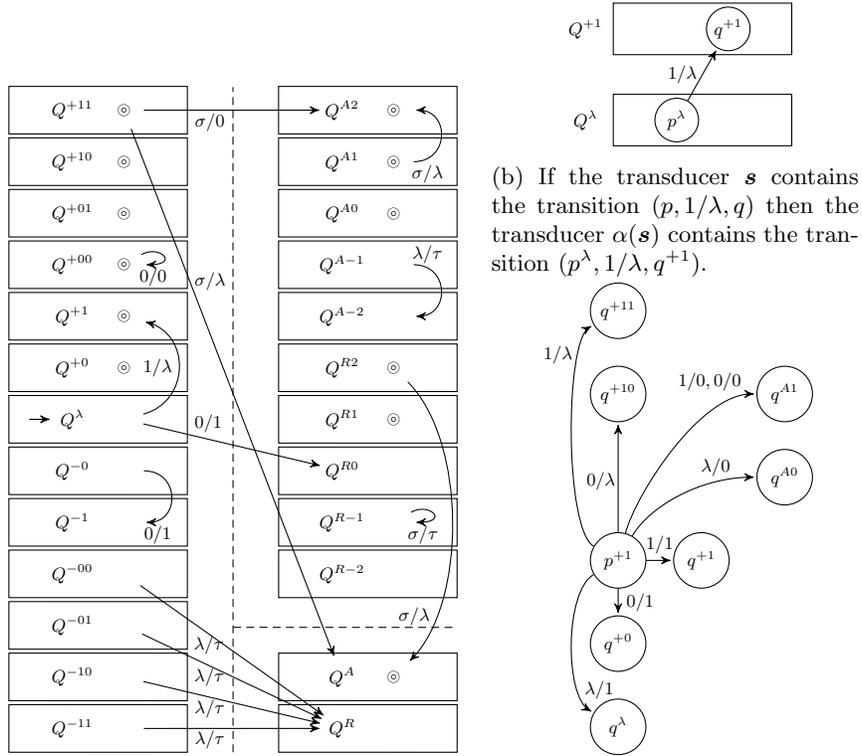

A \emdef{$\ewp$-free transducer} is a transducer that has no label $\ewp$. Using tools from automata theory, we have that every transducer realizes the same relation as one of a $\ewp$-free transducer.

\begin{construction}\label{con:zero}\NOTE{con:zero}
Let $\trs = (Q,\Sigma, T, I, F)$ be a $\ewp$-free and zero-avoiding transducer with some bound $k$. 
The transducer $\alpha(\trs)=(Q',\Sigma,T',I',F')$ is a $C$-copy of $\trs$ as follows. 
\begin{equation}\label{eq:copies}\NOTE{eq:copies}
C=\{\ew,A,R\}\;\cup\;\{+u,-u\mid u\in\alstar,\,1\le|u|\le k \}\;\cup\;\{A\ell,R\ell\mid \ell\in\Z,\,-k\le\ell\le k \}.
\end{equation}
We have $Q'=\cup_{c\in C}Q^c$, $I'=I^\ew$,
\begin{equation}\label{eq:finals}\NOTE{eq:finals}
F' = F^A\cup F^{A0}\cup\big(\cup_{1\le|u|\le k}F^{+u}\big)\cup\big(\cup_{\ell=1}^{k}(F^{A\ell}\cup F^{R\ell})\big)
\end{equation}
The set $T'$ of transitions is defined next. More specifically, for each transition $(p,x/y,q)\in T$, with $x/y\in\{\sigma/\tau,\sigma/\ew,\ew/\tau\mid \sigma,\tau\in \al\}$, we define the set $\Corr{p,x/y,q}$. For each state $p^c\in Q'$ the transition $(p^c,x/y,q^d)$ is in $\Corr{p,x/y,q}$, where $q^d$ depends on $p^c$ and $x/y$ as follows.
\pmsi
If $p^c=p^{\ew}$:
\pnsi\quad if $x/y=\sigma/\sigma$ then $q^d=q^{\ew}$;
\pnsi\quad if $x/y=\sigma/\tau$ and $\sigma>_r\tau$ then $q^d=q^{A0}$;
\pnsi\quad if $x/y=\sigma/\tau$ and $\sigma<_r\tau$ then $q^d=q^{R0}$;
\pnsi\quad if $x/y=\sigma/\ew$, then $q^d=q^{+\sigma}$ if $k>0$, and $q^d=q^A$ if $k=0$;
\pnsi\quad if $x/y=\ew/\tau$, then $q^d=q^{-\tau}$ if $k>0$, and $q^d=q^R$ if $k=0$.
\pmsi
If $p^c=p^{+u}$:
\pnsi\quad if $x/y=\sigma/\ew$ and $|u|<k$ then $q^d=q^{+u\sigma}$;
\pnsi\quad if $x/y=\sigma/\ew$ and $|u|=k$ then $q^d=q^{A}$;
\pnsi\quad if $x/y=\ew/\tau$ and $u[0]=\tau$ then $q^d=q^{u[1..]}$;
\pnsi\quad if $x/y=\ew/\tau$ and $u[0]>_r\tau$ then $q^d=q^{A\ell}$ where $\ell=|u[1..]|$;
\pnsi\quad if $x/y=\ew/\tau$ and $u[0]<_r\tau$ then $q^d=q^{R\ell}$ where $\ell=|u[1..]|$;
\pnsi\quad if $x/y=\sigma/\tau$ and $u[0]=\tau$ then $q^d=q^{+u[1..]\sigma}$;
\pnsi\quad if $x/y=\sigma/\tau$ and $u[0]>_r\tau$ then $q^d=q^{A\ell}$ where $\ell=|u|$;
\pnsi\quad if $x/y=\sigma/\tau$ and $u[0]<_r\tau$ then $q^d=q^{R\ell}$ where $\ell=|u|$.
\pmsi
If $p^c=p^{-u}$:
\pnsi\quad if $x/y=\sigma/\ew$ and $u[0]=\sigma$ then $q^d=q^{-u[1..]}$;
\pnsi\quad if $x/y=\sigma/\ew$ and $u[0]>_r\sigma$ then $q^d=q^{R\ell}$ where $\ell=|u[1..]|$;
\pnsi\quad if $x/y=\sigma/\ew$ and $u[0]<_r\sigma$ then $q^d=q^{A\ell}$ where $\ell=|u[1..]|$;
\pnsi\quad if $x/y=\ew/\tau$ and $|u|<k$ then $q^d=q^{-u\tau}$;
\pnsi\quad if $x/y=\ew/\tau$ and $|u|=k$ then $q^d=q^{R}$;
\pnsi\quad if $x/y=\sigma/\tau$ and $u[0]=\sigma$ then $q^d=q^{-u[1..]\tau}$;
\pnsi\quad if $x/y=\sigma/\tau$ and $u[0]>_r\sigma$ then $q^d=q^{R\ell}$ where $\ell=|u|$;
\pnsi\quad if $x/y=\sigma/\tau$ and $u[0]<_r\sigma$ then $q^d=q^{A\ell}$ where $\ell=|u|$.
\pmsi
If $p^c=p^{X\ell}$ with $X\in\{A,R\}$:
\pnsi\quad if $x/y=\sigma/\ew$ and $\ell<k$ then $q^d=q^{X(\ell+1)}$;
\pnsi\quad if $x/y=\sigma/\ew$ and $\ell=k$ then $q^d=q^{A}$;
\pnsi\quad if $x/y=\ew/\tau$ and $\ell>-k$ then $q^d=q^{X(\ell-1)}$;
\pnsi\quad if $x/y=\ew/\tau$ and $\ell=-k$ then $q^d=q^{R}$;
\pnsi\quad if $x/y=\sigma/\tau$ then $q^d=q^{X\ell}$.
\pmsi
If $p^c\in\{p^A,p^R\}$: $q^d = q^c$.
\hfill$\Box$
\end{construction}

\begin{remark}
	The transitions of $\alpha(\trs)$ have been defined so that the meaning of the states is preserved. Note that any transition of $\alpha(\trs)$ with source state $p^A$ has a destination state of the form $q^A$. This is because both \trs and $\alpha(\trs)$ are zero-avoiding with bound $k$, so any computation $P'$ of $\alpha(\trs)$ ending at $p^A$ has $d_{max}(P')>k$ and $d(P')>0$ and, moreover, any computation $Q'$ of $\alpha(\trs)$ having $P'$ as prefix will be such that $d_{max}(Q')>k$ and  $d(Q')>0$. For similar reasons, any transition of $\alpha(\trs)$  with source state $p^R$ has a destination state of the form~$q^R$.
\end{remark}

\begin{lemma}\label{lem:qs}\NOTE{lem:qs}
	For every computation $P$ of $\trs$, the set $\Corr{P}$ is nonempty. 
\end{lemma}
\begin{proof}
	The statement follows from the  definition of the set of transitions $T'$ of $\alpha(\trs)$. More specifically, let $P=\walk{q_{i-1},x_i/y_i,q_i}_{i=1}^\ell$, and for any $m$ with $1\le m\le\ell$, let $P_m$ be the prefix of $P$ consisting of the first $m$ transitions. Using induction, we show that $\Corr{P_m}\neq\es$ for all $m$. For $m=1$, we have that $\Corr{P_1}$ is nonempty, as $q_0^\ew$ is a defined state of $\alpha(\trs)$ and also a  transition  $(q_0^\ew,x_1/y_1,q_1^d)$ of $\alpha(\trs)$ corresponding to $(q_0,x_1/y_1,q_1)$ is defined in Construction~\ref{con:zero} for all cases of $x_1/y_1$. Now for some $m<\ell$ suppose that there is $P_m'\in\Corr{P_m}$, so we have
	\[
	P'_m=\walk{q_{i-1}^{c_{i-1}},x_i/y_i,q_i^{c_i}}_{i=1}^m.
	\]
	As $q_m^{c_m}$ is defined and there is  a transition $(q_m^{c_m},x_{m+1}/y_{m+1},q_{m+1}^{c_{m+1}})$ defined in Construction~\ref{con:zero} for all cases of $x_{m+1}/y_{m+1}$, we have that also the path
	\[
	\walk{q_{i-1}^{c_{i-1}},x_i/y_i,q_i^{c_i}}_{i=1}^{m+1}
	\]
	is in $\Corr{P_{m+1}}$.
\end{proof}

\begin{lemma}\label{L:alpha}
	Let \trs be a $\ewp$-free and zero-avoiding transducer. The transducer $\alpha(\trs)$ in Construction~\ref{con:zero} is such that
	\[
	\relb{\alpha(\trs)}=\rel\trs\cap\{u/v:u>_rv\}.
	\]
\end{lemma}
\begin{proof}
	First we show that, for every $\win/\wout\in\rel{\alpha(\trs)}$, it is $\win>_r\wout$ and $\win/\wout\in\rel\trs$. That $\win/\wout\in\rel\trs$ follows from Lemma~\ref{L:faithful}. The  part  $\win>_r\wout$ follows from the meaning of the states of $\alpha(\trs)$ and the definition of its final states. More specifically, if $\win/\wout\in\rel{\alpha(\trs)}$ and $P'\in\aComput{\alpha(\trs)}$  with $\lbl{P'}=\win/\wout$ one shows that $\win>_r\wout$ by considering the possible cases for the last state of $P$. For example, if $P$ ends  at a state $f^{R\ell}$, with $\ell>0$, then this means that $\win=x\sigma y$, $\wout=x\tau z$, $\sigma<_r\tau$ and $|y|-|z|=\ell$, which implies $|\win|>|\wout|$; hence $\win>_r\wout$.
	\par Now we show that, if $\win/\wout\in\rel\trs$ and $\win>_r\wout$, then $\win/\wout\in\rel{\alpha(\trs)}$. Let $P$ be any accepting computation of $\trs$ with label $\win/\wout$. By Lemma~\ref{lem:qs}, there is a computation $P'\in\Corr P$ of $\alpha(\trs)$ that ends at some state $f^c$ with $f\in F$ and $c\in C$. If $f^c\in F'$ then $P'$ is accepting and, therefore, the label $\win/\wout$ of $P'$ is in $\rel{\alpha(\trs)}$. If $f^c\notin F'$ then 
\[
c\in\{\ew,R,R0\}\;\cup\;\{-u\mid u\in\alstar,\,1\le|u|\le k \}\;\cup\;\{A\ell,R\ell\mid \ell\in\Z,\,-k\le \ell<0 \}.
\]
Then, we get a contradiction by showing that $\win<_r \wout$ for any $c$ as above. For example, if $c=A\ell$ or $c=R\ell$, with $\ell<0$, the meaning of the states implies that $|\win|<|\wout|$.
\end{proof}

The below theorem solves the rational asymmetric partition problem for every irreflexive relation realized by some zero-avoiding transducer.

\begin{theorem}\label{th:qs}\NOTE{th:qs}
Let \trs be any input-altering and zero-avoiding transducer with bound $k\in\N_0$. Let  $\trt_1=\alpha(\trs)$ and let $\trt_2=\big(\alpha(\trsinv)\big)^{-1}$, where $\alpha(\trs)$ is the transducer produced in Construction~\ref{con:zero}. The following statements hold true.
\begin{enumerate}
\item $|\alpha(\trs)|=\Theta(|\trs||\al|^{k})$.
\item $\{\rel{\trt_1},\rel{\trt_2}\}$ is a rational asymmetric partition of $\rel{\trs}$.
\end{enumerate} 
\end{theorem}
\begin{proof}
For the \underline{first} statement, equation~\eqref{eq:copies} implies that $\alpha(\trs)$ has 
\[
|C|=3+4k+2\frac{|\al|^{k+1}-1}{|\al|-1}
\]
copies of the states of \trs, which is of magnitude $O(|\al|^{k})$. 
If \trs has $n$ states then $\alpha(\trs)$ has $O(n|\al|^{k})$ states, and if \trs has $m$ transitions then $\alpha(\trs)$ has $O(m|\al|^{k})$ transitions.
\par 
The \underline{second} statement follows when we note that  
\[
\rel{\trt_2}=\rel\trs\cap\{u/v:u<_rv\}
\]
which is a consequence of Lemma~\ref{L:alpha} using standard logical arguments on the sets involved; and that  $\big\{\rel\trs\cap\{u/v:u>_rv\},\;\rel\trs\cap\{u/v:u<_rv\}\big\}$ is a partition of $\rel\trs$ because \trs is input-altering.
\end{proof}

Now we have the following consequence of Theorem~\ref{th:qs} and Proposition~\ref{P:lsyncVSzavoid}.

\begin{corollary}\label{cor:lsyn}\NOTE{cor:lsyn}
	Every left synchronous and every right synchronous irreflexive rational relation has a rational asymmetric partition. 
\end{corollary}


\section{An Unsolved Case and a Variation of the Problem}\label{sec:variations}
Recall that Theorem~\ref{th:qs} solves the rational asymmetric partition problem for any irreflexive relation realized by some zero-avoiding transducer. The main open question is the following.
\pmsi \textit{\textbf{Open Question.} Does there exist any rational irreflexive relation that has no rational asymmetric partition?} 
\pmsn
We also offer a more specific open question. Consider the rational symmetric relation: $R=R_1\cup R_1^{-1}\cup R_2\cup R_2^{-1}$, where
\begin{align*}\label{eq:open}\NOTE{eq:open}
	R_1 &= \{(0^a 1^i 0^j 1^b/0^i 1^c 0^j 1^d)\mid a,b,c,d,i,j\in \N\} \quad\text{and}\quad   \\
	R_2 &= \{(0^a 1^i 0^j 1^b,0^c 1^i 0^d 1^j)\mid a,b,c,d,i,j\in \N\}.
\end{align*}
$R$ is rational because $R_1$ and $R_2$ are rational---see for example the
transducer $\trt_1$ realizing $R_1$ in Fig.~\ref{fig:exampleHard}. The more specific question is the following.
\pmsi \textit{Does there exist a rational asymmetric relation $A$ such that $A\cup A^{-1}=R$ ?}
\pmsn
We  note the following facts about $R$:
\pssi$\blacktriangleright$ $R_1\cap R_2=\{(0^a 1^i 0^j 1^b,0^i 1^i 0^j 1^j)\mid a,b,i,j\in \N\}$ is not rational.
\pssi$\blacktriangleright$ Also non rational are: $R_1\cap R_1^{-1}$, $R_1\cap R_2^{-1}$, $R_2\cap R_1^{-1}$, $R_1\cap R_1^{-1}$, $R_1^{-1}\cap R_2^{-1}$, $R_1\cap R_1^{-1}$.
\pssi$\blacktriangleright$ Also non rational is the intersection $R_1\cap R_1^{-1}\cap R_2\cap R_2^{-1}$.
\pssn The fact that $R$ is realized by some transducer that is not zero-avoiding does not imply that $R$ is realized by no zero-avoiding transducer.
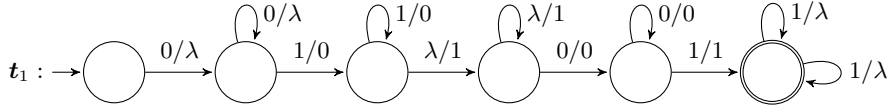
\begin{figure}[hbt]
\begin{center}
\begin{tikzpicture}[>=stealth', initial text={~},shorten >=1pt,auto,node distance=1.75cm,initial text={$\trt_1:$}]
\node[state,initial] (0) {};
\node[state] [right of=0] (1) {};
\node[state] [right of=1] (2) {};
\node[state] [right of=2] (3) {};
\node[state] [right of=3] (4) {};
\node[state,accepting] [right of=4] (5) {};
\draw [->] (0) to  node{$0/\lambda$} (1);
\draw [->] (1) to [loop above] node [near end,right] {$0/\lambda$} (1);
\draw [->] (1) to  node{$1/0$} (2);
\draw [->] (2) to [loop above] node [near end,right] {$1/0$} (2);
\draw [->] (2) to  node{$\lambda/1$} (3);
\draw [->] (3) to [loop above] node [near end,right] {$\lambda/1$} (3);
\draw [->] (3) to  node{$0/0$} (4);
\draw [->] (4) to [loop above] node [near end,right] {$0/0$} (4);
\draw [->] (4) to  node{$1/1$} (5);
\draw [->] (5) to [loop above] node [near end,right] {$1/\lambda$} (5);
\draw [->] (5) to [loop right] node [right] {$1/\lambda$} (5);
\end{tikzpicture}
\end{center}
\caption{Transducer $\trt_1$ realizing the relation $R_1=\{(0^a 1^i 0^j 1^b,0^i 1^c 0^j 1^d)\mid
a,b,c,d,i,j\in \N\}$.}
\label{fig:exampleHard}
\end{figure}

The proof of Theorem~\ref{th:qs} implies that every zero-avoiding and irreflexive-and-symmetric rational relation $S$ has a rational partition according to the radix order $[>_r]$; that is, $\{S\cap[>_r],S\cap[<_r]\}$ is a rational partition of $S$. A question that arises here is whether there are examples of irreflexive-and-symmetric rational  $S$  for which at least one of $S\cap[>_r],S\cap[<_r]$ is not rational. This question can be generalized slightly by using any total asymmetry $[>]$ in place of $[>_r]$. The question would be answered if we find an asymmetric rational $A$ such that at least one of $S\cap[>],S\cap[<]$ is not rational, where $S=A\cup A^{-1}$. 
\begin{description}
  \item[The Rational Non-Partition Problem for a Fixed Asymmetry.] Let $[>]$ be a fixed total asymmetry. Is there an asymmetric rational relation $A$ such that at least one of  $(A\cup A^{-1})\cap[>]$ and $(A\cup A^{-1})\cap[<]$  is not  rational? If the answer is yes, then $A$ is called a \emdef{rational non-partition witness for} $[>]$; else, $A$ is called a \emdef{rational partition witness for} $[>]$.
\end{description}
Next we show that the answer to the above problem is positive for certain fixed total asymmetries $[>]$. We use the following lemma. The notation $\prj1R$ is for the language $\{u:u/v\in R\}$; and the notation $\prj2R$ is for the language $\{v:u/v\in R\}$.

\begin{lemma}
	Let $[>]$ be a total asymmetry and  $A$ be a rational asymmetric relation. If $\prj1A\cap\prj2A=\emptyset$ and exactly one of the languages $\prj1(A\cap[>]),\prj2(A\cap[<])$ is regular then $A$ is a rational non-partition witness for $[>]$.
\end{lemma}
\begin{proof}
	For the sake of contradiction, assume $(A\cup A^{-1})\cap[>]$ is rational. Then, the language $$\prj1(A\cap[>])\cup\prj1(A^{-1}\cap[>])$$ must be regular. This is impossible however, because $\prj1(A^{-1}\cap[>])=\prj2(A\cap[<])$, and the languages $\prj1(A\cap[>]),\prj2(A\cap[<])$ are disjoint and exactly one of the two is regular.
\end{proof}

\begin{proposition}\label{P:variation}\NOTE{P:variation}
	Consider the following asymmetric rational relations.
	\begin{align}
		A &= \{1(00)^j\,/\,00^j1^i\; \mid i,j\in\N_0\}\\
		B &= \{0^{2i+2}101\,/\,0^{i+1}0^j110 \;\mid i,j\in\N_0\}
	\end{align}
	We have that $A$ is a rational non-partition witness for $[>_r]$ and a rational partition witness for $[>_l]$; and $B$ is a rational non-partition witness for both $[>_r]$ and $[>_l]$.
\end{proposition}
\begin{proof}
	We use the above lemma. First for $A$, we have that 
	\pssi\qquad
	$A\cap[>_r]=\{1(00)^j\,/\,00^j1^i\; \mid i,j\in\N_0,\;j\ge i\}\>\text{ and }\>$
	\pssi\qquad
	$A\cap[<_r]=\{1(00)^j\,/\,00^j1^i\; \mid i,j\in\N_0,\;j< i\}.$
	\pssn
	As, $\prj1(A\cap[>_r])$ is regular and $\prj2(A\cap[<_r])$ is not regular, we have that $A$ is a rational non-partition witness for $[>_r]$. On the other hand, as $A\sse[>_l]$ and $A\cap[<_l]=\emptyset$, we have that 
	$$(A\cup A^{-1})\cap[>_l]=A\>\text{ and }\>(A\cup A^{-1})\cap[<_l]=A^{-1},$$
	which implies that $A$ is a rational partition witness for $[>_l]$. Now for $B$, we have that 
	$$B\cap[>_l]=B\cap[>_r]=\{0^{2i+2}101\,/\,0^{i+1}0^j110 \;\mid i,j\in\N_0\,i+1>j\}$$
	and $B\cap[<_l]=B\cap[<_r]=$ the  above set with $i+1\le j$. The statement follows now when we note that $\prj1(B\cap[>_l])$ is regular and $\prj2(B\cap[<_l])$ is not regular.
\end{proof}

\section{Conclusions and Acknowledgement}
Motivated by the embedding problem for rationally independent languages, we have introduced the rational asymmetric partition problem. Our aim was to find the largest class of rational relations that have a rational asymmetric partition. In doing so we introduced zero-avoiding transducers. These define a class of rational relations that properly contain the left and right synchronous relations and  admit  rational asymmetric partitions. Whether all rational relations admit such partitions remains open. We thank Jacques Sakarovitch for looking at this open problem and offering the opinion that it indeed appears to be non trivial.

\if\MINE1\bibliographystyle{plain}
\else \bibliographystyle{splncs03}  
\fi
\bibliography{refs}

\end{document}